\documentclass[11pt,dina4,a4paper, twoside]{amsart}
\usepackage{xcolor,changes}
\usepackage{amsmath,amssymb}
\usepackage{hyperref}
\allowdisplaybreaks

\newcommand{\nn}{\nonumber}

\newcommand{\wt}{\widetilde}

\newcommand{\be}{\begin{equation}}
\newcommand{\ee}{\end{equation}}
\newcommand{\ba}{\begin{array}}
\newcommand{\ea}{\end{array}}
\newcommand{\lag}{\left\langle}
\newcommand{\rag}{\right\rangle}

\newcommand{\bea}{\begin{eqnarray}}
\newcommand{\eea}{\end{eqnarray}}
\newcommand{\beas}{\begin{eqnarray*}}
\newcommand{\eeas}{\end{eqnarray*}}
\newtheorem{theorem}{Theorem}[section]
\newtheorem{lemma}{Lemma}[section]
\newtheorem{proposition}{Proposition}[section]
\newtheorem{corollary}{Corollary}[section]
\newtheorem{remark}{Remark}[section]
\newtheorem{definition}{Definition}[section]

\numberwithin{equation}{section}

\newcommand{\R}{\mathbb{R}}

\newcommand{\C}{\mathbb{C}}

\newcommand{\calH}{{\mathcal H}}
\newcommand{\calM}{{\mathcal M}}
\newcommand{\calT}{{\mathcal T}}
\DeclareMathOperator*{\argmin}{argmin}
\newcommand{\D}{\mathrm{d}}

\newcommand{\Id}{\mathrm{Id}}
\newcommand{\Real}{\mathrm{Re}\,}
\newcommand{\Imag}{\mathrm{Im}\,}

\author{Caroline Lasser}
\address{Zentrum Mathematik, Technische Universit\"{a}t M\"unchen, 85748 Garching bei M\"unchen, Germany}
\email{classer@ma.tum.de}
\author{Chunmei Su}
\address{Yau Mathematical Sciences Center, Tsinghua University, 100084 Beijing, China; Yanqi Lake Beijing Institute of Mathematical Sciences and Applications, Beijing, 101408, China ({\tt sucm@tsinghua.edu.cn})}
\thanks{This work was partially supported by the Alexander von Humboldt Foundation and the I-Site Future programme.}

\title[Various variational approximations of quantum dynamics]{Various variational approximations of quantum dynamics}
\date{\today}
\begin{document}

\maketitle
\begin{abstract}
We investigate variational principles for the approximation of quantum dynamics that apply for approximation manifolds that do not have complex linear tangent spaces. The first one, dating
back to McLachlan (1964) minimizes the residuum of the time-dependent Schr\"odinger equation, while the second one,
originating from the lecture notes of Kramer--Saraceno (1981), imposes the stationarity of an action functional. We characterize both principles in terms of metric and a symplectic orthogonality conditions, consider their conservation properties, and derive an elementary a-posteriori error estimate. As an application, we revisit the time-dependent Hartree approximation and  frozen Gaussian wave packets.
\end{abstract}



\pagestyle{myheadings}\thispagestyle{plain}

\section{Introduction}
We consider an abstract linear Schr\"odinger equation on a complex Hilbert space $\mathcal{H}$ with inner product
$\langle \cdot\, | \,\cdot\rangle$
which is antilinear and linear in its first and second argument, respectively:
\be\label{schro}
i\hbar\dot{\psi}(t)=H\psi(t),
\ee
where $\hbar$ is a reduced Planck's constant, $\dot{\psi}(t)$ represents the time-derivative of the solution $\psi(t)$, and $H:D(H)\to\calH$ is a self-adjoint linear operator on $\mathcal{H}$. For the numerical solution of this Schr\"odinger equation variational principles play an important role. They allow the systematic construction of approximate solutions with guaranteed conservation properties and are successfully applied in many different settings of quantum dynamics \cite{GL,HG,MB}. Let $\mathcal{M}\subseteq \mathcal{H}$ be a smooth submanifold of the Hilbert space on which an approximate solution $u(t)$ to the exact solution $\psi(t)$ is sought,
\[
\psi(t)\approx u(t)\in\calM.
\]
We assume for the initial data that $\psi(0) = u(0)\in\calM$. Our aim here is to analyse several time-dependent variational principles for the construction of such approximations with a special focus on manifolds whose tangent spaces are not complex linear.

\subsection{Dirac--Frenkel variational principle}
The most popular variational principle, the Dirac--Frenkel variational principle \cite[\S26]{DF}, is naturally formulated for approximation manifolds that are compatible with the complex Hilbert space geometry in the following sense. Since $\calM$ is assumed to be a manifold, for any $u\in\mathcal{M}$ there is an associated tangent space $\calT_u\calM$ at $u$. The tangent space at $u$ is a real-linear subspace of the Hilbert space that intuitively contains all possible directions in which one can tangentially pass through $u$. Entering the framework developed by C. Lubich in \cite{Lu05} and \cite[\S II.1]{Lu}, we assume that the tangent spaces are not just real-linear, but complex-linear subspaces of the Hilbert space.
We then determine the function $t\rightarrow u(t)\in \mathcal{M}$ by the requirement that for all $t$,
the derivative $\dot{u}(t)$ lies in the tangent space $\mathcal{T}_{u(t)}\mathcal{M}$ and by imposing the orthogonality condition
\be\label{DF}
\langle v\,| \,\dot{u}(t)-\frac{1}{i\hbar}Hu(t)\rangle=0,\quad \forall v\in \mathcal{T}_{u(t)}\mathcal{M}.
\ee
Equivalently, one may rephrase the orthogonality condition in terms of the
orthogonal projection
$P_{u(t)}:\calH\to\calT_{u(t)}\calM$
onto the tangent space at $u(t)$ and state the Dirac--Frenkel principle as the non-linear evolution equation
\be\label{eq:DF}
i\hbar \dot u(t) = P_{u(t)}Hu(t).
\ee
In this complex-linear setting, without any additional assumptions, the Dirac--Frenkel approximation has two striking properties. First, it satisfies the a-posteriori error estimate
\[
\|\psi(t)-u(t)\| \le \int_0^t \mathrm{dist}(\calT_{u(s)}\calM,\frac{1}{i\hbar}Hu(s)) \,\D s
\]
with respect to the Hilbert space norm $\|\cdot\| = \langle \cdot\,|\,\cdot\rangle^{1/2}$. Second, energy is conserved, that is, expectation values for the Schr\"odinger operator $H$ satisfy
\[
\langle u(t)\,|\, Hu(t)\rangle = \langle\psi(0)\,|\, H\psi(0)\rangle,\quad\forall t.
\]
Moreover, if the tangent spaces satisfy the additional property $u(t)\in\calT_{u(t)}\calM$, then also
the norm is conserved.

\subsection{More variational principles}
For more general situation, where the tangent spaces of the approximation manifold are not complex-linear,
but only real-linear subspaces of the Hilbert space, the literature commonly refers to two variational principles:
the minimum distance variational principle (MVP) introduced by A. McLachlan in \cite{Mc},
and the so-called time-dependent variational principle (TDVP) of P. Kramer and M. Saraceno \cite[\S2.1]{TD}. Both principles require
that $\dot u(t)\in\calT_{u(t)}\calM$ for all time $t$, but then add different additional conditions.
\begin{itemize}
\item[] McLachlan (MVP). $\dot{u}(t)$ is chosen to minimize the residuum,
\[\left\|i\hbar\dot{u}(t)-Hu(t)\right\|=\min\limits_{w\in\mathcal{T}_{u(t)}\mathcal{M}}
\|i\hbar w-Hu(t)\|.\]
\item[] Kramer--Saraceno (TDVP). $\dot{u}(t)$ is chosen such that the action functional
\[F(u)=\int_{t_1}^{t_2}\left\langle u(t)\,|\, i\hbar \dot{u}(t)-Hu(t)\right\rangle \D t\]
satisfies the stationarity condition $\delta F(u) = 0$, where the variation is taken with respect
to paths in $\calM$ with fixed end-points at time $t_1$ and $t_2$.
\end{itemize}
Our aim here is to investigate these two principles, to provide equivalent reformulations, to analyse their conservation properties (Section~\ref{sec:McL} and \ref{sec:KS}), to derive local-in-time error estimates (Section~\ref{sec:error}), interpret the upper bounds in terms of energy fluctutations (Section~\ref{sec:energy}) and to apply these findings to the time-dependent
Hartree approximation and frozen Gaussian wave packets (Section \ref{sec:Hartree} and \ref{secfro}).

\section{The McLachlan variational principle}\label{sec:McL}
The McLachlan principle (MVP) applies to the general situation with real-linear tangent spaces. It determines the time derivative $\dot u(t)$ of the variational solution by requiring that $\dot u(t)\in\calT_{u(t)}\calM$ and
by imposing the  minimization
\[
\min\limits_{w\in\mathcal{T}_{u(t)}\mathcal{M}}\|i\hbar w-Hu(t)\| =
\min\limits_{w\in\mathcal{T}_{u(t)}\mathcal{M}}\hbar\,\|w-\frac{1}{i\hbar}Hu(t)\|,
\]
that is,
\be\label{mvp}
\dot u(t) = \argmin_{w\in\calT_{u(t)}\calM} \|w-\frac{1}{i\hbar}Hu(t)\|.
\ee
The tangent space $\calT_{u(t)}\calM$ is a real-linear subspace
of the complex Hilbert space $\calH$. In particular, it is a closed and convex subset.
Therefore, there exists a unique element
of best approximation
$w_0\in\calT_{u(t)}\calM$ to the vector $\tfrac{1}{i\hbar}Hu(t)$
with respect to the norm of the Hilbert space $\calH$,
see for example \cite[Theorem~3.5]{De}. The McLachlan principle sets $\dot u(t) = w_0$.
The element of best approximation can also be described in terms of an orthogonal projection,
when equipping the normed vector space
$(\calH,\|\cdot\|)$ with the compatible real-valued inner product
\[
g:\ \calH\times\calH\to\R,\quad (v,w)\mapsto{\rm Re}\,\langle v\,|\, w\rangle,
\]
that might also be called a metric.
We emphasize, that the norms induced by the complex and the real inner product are the same,
since $\langle v\,|\, v \rangle = \Real\langle v\,|\,v\rangle$ for all $v\in\calH$. We
denote for each $u\in\calM$ by $P_{u}^g:\calH\to\calT_u\calM$ the orthogonal projection onto
the tangent space that is defined by the real-valued inner product,
\be\label{projection}
\Real\langle v\,|\,P_{u}^gw\rangle = \Real\langle v\,|\,w\rangle\quad\forall v\in\calT_u\calM, \ \forall w\in\calH.
\ee
Using the orthogonal projection, we may write the best approximation
as $w_0 = P_{u(t)}^g\frac{1}{i\hbar}Hu(t)$ and thus $u(t)$ as the solution of the non-linear evolution equation
\[
\dot u(t) = P_{u(t)}^g\,\frac{1}{i\hbar} Hu(t).
\]
In comparsion with the Dirac--Frenkel evolution equation \eqref{eq:DF}, we note the difference in the placement
of the imaginary unit, that cannot be pulled outside the orthogonal projector for the case of real-linear
tangent spaces.
The following Proposition \ref{prop:metric} recasts the McLachlan principle in terms of an orthogonality condition
for the real-valued inner product. Even though our previous considerations qualify for a proof, we provide
an additional explicit minimization argument.

\begin{proposition}[Metric formulation]\label{prop:metric}
The McLachlan variational principle requires that $\dot u(t)\in\calT_{u(t)}\calM$ and
\be\label{metric}
\mathrm{Re}\langle v\,| \,\dot{u}(t)-\frac{1}{i\hbar}Hu(t)\rangle=0,\quad \forall v\in \mathcal{T}_{u(t)}\mathcal{M}.
\ee
\end{proposition}

\begin{proof}
We first mention that the local minima of the function
\[
f: \calT_{u(t)}\calM\to[0,\infty), \quad f(w) = \|w - \frac{1}{i\hbar}Hu(t)\|^2
\]
are global ones, since the norm $\|\cdot\|$ is induced from an inner product and thus strictly convex.
We notice that for all $v\in\calT_{u(t)}\calM$ and $\tau\in\R$,
\begin{align*}
f(\dot u(t)+\tau v) &= \|(\dot{u}(t)+\tau v)-\frac{Hu(t)}{i\hbar}\|^2\\
 &=f(\dot u(t))+2\tau\,\mathrm{Re}\langle v\,| \,\dot{u}(t)-\frac{Hu(t)}{i\hbar}\rangle+\tau^2\|v\|^2.
\end{align*}
Therefore, the first and second directional derivatives of $f$ in $\dot u(t)$ satisfy
\begin{align*}
\delta f(\dot u(t),v) &= \lim_{\tau\to0} \frac{f(\dot u(t)+\tau v)-f(\dot u(t))}{\tau}= 2\mathrm{Re}\langle v\,| \,\dot{u}(t)-\frac{1}{i\hbar}Hu(t)\rangle
\end{align*}
and
\[
\delta^2 f(\dot u(t),v)= \lim_{\tau\to0} \frac{f(\dot u(t)+\tau v)-2f(\dot u(t))+f(\dot u(t)-\tau v)}{\tau^2}= 2\|v\|^2.\]
Since the second variation is strictly positive, $\dot u(t)$ is a local minimizer of $f$ if and only if
$f$ is stationary in $\dot u(t)$, that is, if and only if \eqref{metric} holds.
\end{proof}

The above metric reformulation of the McLachlan principle immediately implies its equivalence with the
Dirac--Frenkel principle for manifolds with tangent spaces that are complex linear subspaces of the Hilbert
space. It also allows us to repeat the elementary proof
for norm conservation of the Dirac--Frenkel principle, see also \cite[Theorem~1.4]{Lu}.

\begin{lemma}[Norm conservation]
If $u(t)\in\calT_{u(t)}\calM$ for all $t$, then the McLachlan principle is norm-conserving.
\end{lemma}

\begin{proof}
Taking $v=u(t)$ in \eqref{metric} yields:
\[
\frac{\D}{\D t}\|u(t)\|^2 = 2\,\mathrm{Re}\langle u(t)\,|\,\dot{u}(t)\rangle
=
2 \mathrm{Re}\langle u(t)\,|\,\frac{1}{i\hbar}Hu(t)\rangle=0,
\]
where the last equation uses that $H$ is self-adjoint.
\end{proof}

\section{The Kramer--Saraceno variational principle}\label{sec:KS}

The Kramer--Saraceno principle (TDVP) also applies to the general case with real-linear tangent spaces. It determines the variational solution $u(t)$ as the stationary point of the action functional
\[
F(u)=\int_{t_1}^{t_2}\left\langle u(t)\,|\, i\hbar \dot{u}(t)-Hu(t)\right\rangle \D t,
\]
where the variation is done with respect to fixed end point conditions.
In analogy to our previous considerations that used the metric induced by
the complex inner product, we now use the symplectic form
\[
\omega:\ \calH\times\calH\to\R,\quad(v,w)\mapsto\Imag\langle v\,|\,w\rangle,
\]
that stems from the imaginary part. We characterize the principle
by a symplectic orthogonality condition:

\begin{proposition}[Symplectic formulation] The Kramer--Saraceno variational principle requires that
$\dot u(t)\in\calT_{u(t)}\calM$ and
\be\label{symplectic}
\mathrm{Im}\langle v\,| \,\dot{u}(t)-\frac{1}{i\hbar}Hu(t)\rangle=0,\quad \forall v\in \mathcal{T}_{u(t)}\mathcal{M}.\ee
\end{proposition}

\begin{proof}
Let $(u_\tau(t))_{\tau}$ be a path in $\calM$ with a real parameter $\tau$ that satisfies
$u_0(t) = u(t)$ together with the end-point conditions
\[
u_\tau(t_1) = u(t_1), \quad u_\tau(t_2) = u(t_2)\quad\text{for all}\  \tau.
\]
We differentiate with respect to $\tau$ and obtain
\[
\frac{\D}{\D\tau} F(u_\tau) = \int_{t_1}^{t_2}
\left(\left\langle \partial_\tau u_\tau\,|\, i\hbar \dot u_\tau-Hu_\tau\right\rangle +
\left\langle u_\tau\,|\, i\hbar \partial_\tau\dot u_\tau-H\partial_\tau u_\tau\right\rangle\right) \D t.
\]
Since $H$ is self-adjoint, we have
\[
\int_{t_1}^{t_2}
\left(-\left\langle \partial_\tau u_\tau\,|\,Hu_\tau\right\rangle -
\left\langle u_\tau\,|\, H\partial_\tau u_\tau\right\rangle\right) \D t =
-2\int_{t_1}^{t_2}
{\rm Re}\left\langle \partial_\tau u_\tau\,|\,Hu_\tau\right\rangle \D t.
\]
An integration by parts provides
\[
\int_{t_1}^{t_2}
\left(\left\langle \partial_\tau u_\tau\,|\, i\hbar \dot u_\tau\right\rangle +
\left\langle u_\tau\,|\, i\hbar \partial_\tau\dot u_\tau\right\rangle\right) \D t =
2\int_{t_1}^{t_2}
{\rm Re}\left\langle \partial_\tau u_\tau\,|\, i\hbar \dot u_\tau\right\rangle \D t.
\]
Therefore,
\begin{align*}
\frac{\D}{\D\tau} F(u_\tau)
&= 2\int_{t_1}^{t_2}{\rm Re}\left\langle \partial_\tau u_\tau\,|\, i\hbar \dot u_\tau-Hu_\tau\right\rangle \D t\\
&= -2\hbar\int_{t_1}^{t_2}{\rm Im}\langle\, \partial_\tau u_\tau\,|\,\dot u_\tau-\frac{1}{i\hbar} Hu_\tau\rangle\, \D t.
\end{align*}
Thus the stationarity condition
$\frac{\D}{\D\tau} F(u_\tau)\Big|_{\tau = 0} = 0$
yields \eqref{symplectic}.
\end{proof}

For each $u\in\calM$ we denote by $P_u^\omega:\calH\to\calT_u\calM$ the symplectic projection onto the tangent space, that is defined by the condition
\[
\Imag\langle v\,|\,P_{u}^\omega w\rangle = \Imag\langle v\,|\,w\rangle\quad\forall v\in\calT_u\calM, \ \forall w\in\calH.
\]
In terms of the symplectic projection, we may combine the requirement $\dot u(t)\in \calT_{u(t)}\calM$ and the symplectic orthogonality condition~\eqref{symplectic} and equivalently express them as the non-linear evolution equation
\[
\dot u(t) = P^\omega_{u(t)}\frac{1}{i\hbar}Hu(t).
\]
For the special case that all the tangent spaces are complex linear subspaces, all three projections coincide, that is,
$P_u = P_u^g = P_u^\omega$, and all three variational principles (Dirac--Frenkel, McLachlan, Kramer--Saraceno) are equivalent.

Similarly to the McLachlan principle, the Kramer--Saraceno principle can be norm-conserving, if the manifold's tangent spaces satisfy an additional property. In contrast to the McLachlan principle, the elementary argument
for energy-conservation given for the Dirac--Frenkel principle in \cite[Theorem~1.1]{Lu} also applies
to the Kramer--Saraceno principle.

\begin{lemma}[Conservation properties]\label{lem:cons_sympl}
The Kramer--Saraceno principle satisfies the following conservation properties:
\begin{enumerate}
\setlength{\itemsep}{1ex}
\item
If $iu(t)\in\calT_{u(t)}\calM$, then the Kramer--Saraceno principle is norm-conserving.
\item
The Kramer--Saraceno principle is energy conserving.
\end{enumerate}
\end{lemma}

\begin{proof}
We first observe that the symplectic orthogonality condition~\eqref{symplectic} can equivalently be restated as
\be\label{symplectic_real}
\mathrm{Re}\langle v\,| \,i\hbar\dot{u}(t)-Hu(t)\rangle=0,\quad \forall v\in \mathcal{T}_{u(t)}\mathcal{M}.
\ee
We now consider norm conservation. We take $v=iu(t)$ in \eqref{symplectic_real} and obtain:
\begin{align*}
\frac{\D}{\D t}\|u(t)\|^2 &= 2\,\mathrm{Re}\langle u(t)\,|\,\dot{u}(t)\rangle
= 2\,\mathrm{Re}\langle iu(t)\,|\,i\dot{u}(t)\rangle =
2\, \mathrm{Re}\langle iu(t)\,|\,\frac{1}{\hbar}Hu(t)\rangle=0,
\end{align*}
where the last equation uses that $H$ is self-adjoint. For energy conservation, we take
$v=\dot{u}(t)$ in \eqref{symplectic_real} and immediately derive
\[
\frac{\D}{\D t}\langle u(t)\,|\,Hu(t)\rangle=2\,\mathrm{Re}\langle \dot{u}(t)\,|\,Hu(t)\rangle=
2\,\mathrm{Re}\langle \dot{u}(t)\,|\,i\hbar\dot{u}(t)\rangle=0,
\]
where the first equation uses the self-adjointness of $H$.
\end{proof}

The following Table~\ref{tab} summarizes the conservation properties and the non-linear evolution equations of both variational principles.

\begin{table}[h]
\caption{Conservation and evolution properties of the McLachlan and the Kramer--Saraceno variational principle.}\label{tab}
\begin{center}
\begin{tabular}{c|c|c}
& McLachlan & Kramer--Saraceno\\*[1ex]\hline
& & \\*[-2ex]
norm & if $u(t)\in\calT_{u(t)}\calM$ & if $iu(t)\in\calT_{u(t)}\calM$\\*[2ex]
energy & not guaranteed & yes\\*[2ex]
$\dot u(t) = $ & $P^g_{u(t)}\frac{1}{i\hbar} Hu(t)$ & $P^\omega_{u(t)}\frac{1}{i\hbar} Hu(t)$\\*[0ex]
\end{tabular}
\end{center}
\end{table}

\section{A posteriori error bounds}\label{sec:error}
Our previous discussion of the McLachlan and the Kramer--Saraceno priniciple suggests to define the action of a variational Hamiltonian on the approximate solution $u(t)\in\calM$ according to
\be\label{proj}
H_{u(t)}u(t) = \left\{ \begin{array}{ll} i P^g_{u(t)}\frac{1}{i} Hu(t) & \text{(McLachlan principle)}\\*[1ex]
i P^\omega_{u(t)}\frac{1}{i} Hu(t) & \text{(Kramer--Saraceno principle),}
\end{array}\right.
\ee
and to write the variational equations of motion as
\[
i\hbar\dot u(t) = H_{u(t)}u(t).
\]
We then define
\be\label{unif}
\varepsilon_\calM(u(t)) = \frac{1}{\hbar}\|H_{u(t)}u(t)-Hu(t)\| = \|\dot{u}(t)-\frac{1}{i\hbar}Hu(t)\|.
\ee
Recalling the characterization \eqref{mvp} of the McLachlan principle, that is,
\[
\dot u(t) =
\argmin_{w\in \mathcal{T}_{u(t)}\mathcal{M}}\|w-\frac{1}{i\hbar}Hu(t)\|,
\]
we see that for the McLachlan solution, $\varepsilon_\calM(u(t))$ measures the minimal distance between the tangent space at $u(t)$ and $\frac{1}{i\hbar}Hu(t)$.
We explore the relation of this residual measure to the a posteriori error. Our result is analogous to the one in \cite[Section~2.2]{Lu05} and \cite[Theorem~1.5]{Lu}, where
the estimate was proven for the Dirac--Frenkel principle for manifolds $\calM$ with complex linear tangent spaces.

\begin{theorem}[A posteriori error]\label{theorem}
If $u(0)=\psi(0)\in\mathcal{M}$, then the McLachlan and the Kramer--Saraceno variational principle yield the a posteriori error bound
\be\label{eror}
\|u(t)-\psi(t)\|\le \int_0^t\varepsilon_\mathcal{M}(u(s))\, \D s.
\ee
\end{theorem}

\begin{proof}
Denote the error as $e(t):=u(t)-\psi(t)$. We have
\[
\dot e(t) = \frac{1}{i\hbar}H e(t) + f(t),\quad f(t) =\dot{u}(t)-\frac{1}{i\hbar}Hu(t).\]
We observe that $\|f(t)\| = \varepsilon_\calM(u(t))$ for all $t$.
By the variation of constants formula, we obtain
\[
e(t) = \int_0^t {\rm e}^{-iH(t-s)/\hbar} f(s) \, \D s,
\]
and by unitarity
\[
\|e(t)\| \le \int_0^t \|f(s)\|\, \D s=\int_0^t\varepsilon_\calM(u(s))\D s,
\]
which completes the proof.
\end{proof}

For the McLachlan principle, the minimal distance $\varepsilon_\calM(u)$ allows for an alternative representation,
that is more practical for numerical simulation.

\begin{lemma}[Minimal distance]\label{min} For the McLachlan principle, we have
\be\label{epn}
\varepsilon_\mathcal{M}^2(u)= \frac{1}{\hbar^2}\|Hu\|^2-\|\dot{u}\|^2.
\ee
In particular, $\hbar\|\dot u\|\le \|Hu\|$.
\end{lemma}

\begin{proof} Taking $v=\dot{u}$ in \eqref{metric}, we get $\|\dot{u}\|^2=\mathrm{Re}\langle \dot{u}\,|\, \frac{1}{i\hbar}Hu\rangle$, which then implies
\begin{align*}
\varepsilon^2_\mathcal{M}(u)&=\|\dot{u}-\frac{1}{i\hbar}Hu\|^2=\|\dot{u}\|^2- 2\mathrm{Re}\langle \dot{u}\,|\,\frac{1}{i\hbar}Hu\rangle+\frac{1}{\hbar^2}\|Hu\|^2\\
&=\frac{1}{\hbar^2}\|Hu\|^2-\|\dot{u}\|^2,
\end{align*}
and the result is concluded.
\end{proof}

\section{Energy fluctuations}\label{sec:energy}

The obtained evolution equation for the McLachlan principle can be cast in the usual Schr\"odinger type format
\[
i\hbar\dot u(t) = H_{u(t)} u(t)
\]
where the variational Hamiltonian $H_{u(t)}$, that acts linearly on the Hilbert space $\calH$ but depends non-linearly on the variational solution $u(t)$, is given by
\[
H_{u(t)} = i \, P_{u(t)}^g \frac{1}{i}H
\]
The orthogonal projection $P_{u(t)}^g$ is defined with respect to the real-linear inner product
$g = \mathrm{Re}\,\langle\cdot\,|\,\cdot\rangle$. For the special case of complex-linear tangent spaces
it is natural to work with the full complex inner product $\langle\cdot\,|\,\cdot\rangle$, the Dirac--Frenkel principle and the complex-linear orthogonal projection, which then supersedes the double placement of the imaginary unit $i$.
Following the interpretation of R. Martinazzo and I. Burghardt in \cite{MB}, we next express the minimal distance
$\varepsilon_\calM(u(t)) = \|Hu(t)-H_{u(t)}u(t)\|/\hbar$ of the variational solution $u(t)$ in terms of energies and energy fluctuations, which are key concepts for the statistical interpretation of quantum mechanics.

\medskip
\begin{proposition}[Energy fluctuation]\label{pro1}
Assume that the McLachlan variational solution $u(t)$ is normalised $\|u(t)\|=1$.
Then, the minimal distance satisfies
\begin{align}\label{ener}
\varepsilon_\calM^2(u) &= \frac{1}{\hbar^2}
\left(\|(H-E_0)u\|^2-\|(H_{u}-E_{u})u\|^2\right) + \delta_{\rm en}(u)
\end{align}
for all time $t$, where $\delta_{\rm en}(u) = \frac{1}{\hbar^2} \left(E_0^2 -|E_{u}|^2\right)$ is the scaled difference of the squares of the energy for the Hamiltonian $H$, that is, $E_0(t) = \lag u(t)\,|\,Hu(t)\rag$, and the energy for the variational Hamiltonian $H_u$, that is, \[E_{u(t)} = \lag u(t)\,|\,H_{u(t)}u(t)\rag.\]

\end{proposition}

\begin{proof}
We firstly observe that
\[
\|(H-E_0)u\|^2= \|Hu\|^2 - 2\,{\rm Re}\langle Hu\,|\, E_0 u\rangle + |E_0|^2\,\|u\|^2= \|Hu\|^2 - E_0^2,
\]
where we have used that $\langle u\,|\,Hu\rangle$ is a real number and that $u$ is normalized.
The same consideration applies for the variational Hamiltonian $H_u$ and yields
\[
\|(H_u-E_u)u\|^2 = \|H_uu\|^2 - |E_u|^2.
\]
Thus, by Lemma~\ref{min}, we have
\begin{align*}
\varepsilon_\calM^2(u) &=\frac{1}{\hbar^2}\left(\|Hu\|^2-\|H_u u\|^2\right)\\
&=\frac{1}{\hbar^2}
\left( \|(H-E_0)u\|^2- \|(H_u-E_u)u\|^2\right) +
\frac{1}{\hbar^2}\big( E_0^2 - |E_u|^2\big),
\end{align*}
which completes the proof.
\end{proof}

If the approximation manifold has properties similar to the
one chosen in Section~\ref{sec:Hartree} for the time-dependent Hartree approximation
or to the frozen Gaussians' manifold considered in Section~\ref{secfro},
then the minimal distance is a mere difference between energy fluctuations.
\smallskip

\begin{corollary}\label{coro}
Assume that the McLachlan variational solution $u(t)$ satisfies $\|u(t)\|=1$, $iu(t)\in \calT_{u(t)}\calM$.
Then, the energy expectation values for $H$ and $H_u$ coincide, $E_{u(t)} = E_0(t)$
for all time $t$, and the minimal distance is solely the difference of the energy fluctuations
\[
\varepsilon_\calM^2(u)= \frac{1}{\hbar^2}
\left(\|(H-E_0)u\|^2-\|(H_u-E_0)u\|^2\right).\]
\end{corollary}

\begin{proof} Norm conservation implies $\mathrm{Re}\lag u\,|\,\dot{u}\rag=0$ and therefore also
\[\mathrm{Im}\lag u\,|\,i\hbar\dot{u}-Hu\rag=0,\]
since $\lag u\,|\,Hu\rag\in\mathbb{R}$.
On the other hand, since $iu\in \calT_{u} \calM$, the McLachlan variational principle gives $\mathrm{Im}\lag iu\,|\, i\hbar \dot{u}-Hu\rag=0$, which implies 
\[\mathrm{Re}\lag u\,|\,i\hbar\dot{u}-Hu\rag=0.\] Thus we get
\be\label{eq:cond}
\lag u\,|\, i\hbar \dot u-Hu\rag = 0,
\ee
so that the energy expectation values coincide, i.e., 
\[\langle u\,|\, H_uu\rangle=\lag u\,|\, i\hbar\dot{u}\rag=\langle u\,|\, Hu\rangle.\]
\end{proof}
\smallskip
\begin{remark}
For the case of complex linear tangent spaces $\calT_u\calM$ that contain $u$ for all $u\in\calM$, the Dirac--Frenkel variational solution (initialized with normalized data $u(0)=\psi_0$) naturally fulfills the conditions in Corollary \ref{coro}, so that  in this case, the minimal distance is a mere difference of energy fluctuations.
\end{remark}


\section{Time-dependent Hartree approximation}\label{sec:Hartree}
Consider the Schr\"odinger equation \eqref{schro} with Hamiltonian
\be\label{Ham}
H=\sum_{n=1}^N h_n+V
\ee
where each $h_n$ represents a single-particle Hamiltonian acting on the variable $x_n\in\mathbb{R}^3$,
and $V:=V(x): \mathbb{R}^{3N}\rightarrow \mathbb{R}$ is the interaction potential
that couples all variables in $x=(x_1, \ldots, x_N)$.
The $n$-th single-particle Hamiltonian could be, but need not be the Laplacian $h_n = -\frac{\hbar^2}{2M_n}\Delta_{x_n}$.
Formulating and analysing the time-dependent Hartree approximation requires the inner
products of $L^2(\R^{3N})$, $L^2(\R^3)$, and $L^2(\R^{3(N-1)})$, which we all denote by $\lag\cdot\,|\,\cdot\rag$. The context will clarify, whether the integration is with respect to $x$, one $x_n$ or all variables except one $x_n$.

\subsection{The Hartree manifold}
We consider the Hartree manifold
\[\mathcal{M}=\Big\{u\in L^2(\mathbb{R}^{3N})\mid\, u(x)=\prod\limits_{n=1}^N\varphi_n(x_n),\,\,\, \varphi_n\in L^2(\mathbb{R}^3),\,\,\, \lag\varphi_n\,|\,\varphi_n\rag=1\Big\},\]
that consists of tensor products of normalized square integrable functions, that are called the single-particle functions.
We note that $u\in\calM$ satisfies $\|u\| = 1$ and that its representation as $u(x)=\varphi_1(x_1)\varphi_2(x_2)\cdots\varphi_N(x_N)$ is not unique: for any numbers $a_n\in\mathbb{C}$ satisfying $|a_n|=1$ and $a_1\cdots a_N = 1$, $u$ remains unaltered under the transformation $\varphi_n\rightarrow a_n\varphi_n$, $n=1, \ldots, N$.

\begin{lemma}[Tangent space]\label{Tud}
For given $u=\prod\limits_{n=1}^N\varphi_n\in \calM$, we denote by
$\psi_n=\prod\limits_{m\neq n}\varphi_m\in L^2(\R^{3(N-1)})$
the corresponding $n$-th single-hole wave function. Then, any tangent function $v\in\calT_u\calM$ has a unique representation of the form
$v=\sum\limits_{n=1}^Nv_n\psi_n$
with square integrable functions $v_n\in L^2(\mathbb{R}^3)$ satisfying the gauge conditions
\be\label{eq:gauge}
\mathrm{Re}\lag v_1\,|\,\varphi_1\rag = 0\quad\text{and}\quad
\lag v_n\,|\,\varphi_n\rag = 0\quad \forall n=2,\ldots,N.
\ee
In particular, the tangent space $\calT_u\calM$ at $u$ is not complex-linear.
We have $iu\in\calT_u\calM$ but $u\not\in\calT_u\calM$.
\end{lemma}

\begin{proof} Suppose $g(s)=\prod\limits_{n=1}^N f_n(x_n, s)\in \calM$ is a path passing through $u$ with $f_n(\cdot,s)\in L^2(\R^3)$, $\lag f_n(\cdot, s)\,|\,f_n(\cdot, s)
\rag=1$ and $f_n(x_n, 0)=\varphi_n(x_n)$. We find that
\[\dot g(s)=\sum\limits_{n=1}^Nf_1(x_1, s)\cdots f_{n-1}(x_{n-1}, s)\dot{f}_n(x_n, s)f_{n+1}(x_{n+1}, s)\cdots f_N(x_N, s),\]
which gives
$\dot g(0)=\sum\limits_{n=1}^N\dot{f}_n(x_n, 0)\psi_n$. Due to the normalisation we also have
\[\mathrm{Re}\lag \dot{f}_n(\cdot, s)\,|\,f_n(\cdot, s)
\rag=\frac{1}{2}\frac{\D}{\D s}\lag f_n(\cdot, s)\,|\,f_n(\cdot, s)
\rag=0.\]
Taking $s=0$ gives $\mathrm{Re}\lag \dot{f}_n(\cdot, 0)\,|\,\varphi_n(\cdot)\rag = 0$,
and thus for any $v\in\calT_u\calM$  the existence of a representation
$v = \sum\limits_{n=1}^N v_n\psi_n$
such that $\mathrm{Re}\lag v_n\,|\,\varphi_n\rag = 0$ for all $n$. Such a representation is
not unique, since we can choose any $\alpha\in\R^N$ with $\alpha_1+\cdots+\alpha_N = 0$ and
set $\wt v_n = v_n-i\alpha_n\varphi_n$. Then,
\[
\sum_{n=1}^N \wt v_n\psi_n = \sum_{n=1}^N v_n\psi_n - i\big(\sum_{n=1}^N\alpha_n\big)u = v
\]
and $\mathrm{Re}\lag\wt v_n\,|\,\varphi_n\rag = \mathrm{Re}\lag i\alpha_n\varphi_n\,|\,\varphi_n\rag = 0$ for all $n$. However, the particular choice
\[
\alpha_n = -\mathrm{Im}\lag v_n\,|\,\varphi_n\rag\quad\forall n=2,\ldots,N\quad\text{and}\quad
\alpha_1 = -(\alpha_2+\cdots+\alpha_N)
\]
provides a representation of the tangent function $v$ that satisfies the claimed gauge condition \eqref{eq:gauge}.
On the other hand, let us consider
\[
v=\sum\limits_{n=1}^N v_n\psi_n=\sum\limits_{n=1}^N \widetilde{v}_n\psi_n
\in\calT_u\calM\]
with $(v_n)_{n=1}^N$ and $(\wt v_n)_{n=1}^N$ satisfying \eqref{eq:gauge}. Then, for any $n=1,\ldots,N$ and any function $\vartheta\in L^2(\R^3)$, the
tensor product $\varphi_1(x_1)\cdots \vartheta(x_n)\cdots\varphi_N(x_N)$ that carries~$\vartheta$ as the
$n$-th factor satisfies
\begin{align*}
\lag \varphi_1\cdots\vartheta\cdots\varphi_N\,|\,v\rag &= \sum\limits_{k=1}^N
\lag \varphi_1\cdots\vartheta\cdots\varphi_N\,|\,v_k\psi_k\rag\\
&=\lag \vartheta\,|\,\varphi_n\rag\sum\limits_{k\neq n}
\lag \varphi_k\,|\,v_k\rag  + \lag \vartheta\,|\,v_n\rag\\
&=\lag \vartheta\,|\,\varphi_n\rag\sum\limits_{k\neq n}
\lag \varphi_k\,|\,\wt{v}_k\rag  + \lag \vartheta\,|\,\wt{v}_n\rag.
\end{align*}
For $n=1$, this implies $\lag \vartheta\,|\,v_1\rag=\lag \vartheta\,|\,\wt{v}_1\rag$ and thus $v_1=\wt v_1$.
For $n\ge 2$, the above relation means
\[
\lag \vartheta\,|\,\varphi_n\rag\lag \varphi_1\,|\,v_1\rag +
\lag \vartheta\,|\,v_n\rag= \lag \vartheta\,|\,\varphi_n\rag\lag \varphi_1\,|\,\wt v_1\rag +
\lag \vartheta\,|\,\wt v_n\rag,
\]
and we obtain $v_n=\wt v_n$ for all $n$. As for the representation of $iu$, we choose
$v_1 = i\varphi_1$ and $v_2 = \cdots = v_N = 0$ to see that $iu\in\calT_u\calM$.
\end{proof}

\subsection{The variational equations of motion}
We consider the McLachlan and the Kramer--Saraceno variational principle on the Hartree manifold~$\calM$, i.e., we find
\[
u(x, t)=\varphi_1(x_1, t)\cdots\varphi_N(x_N,t)\in\calM
\]
with each
$\lag \varphi_n\,|\,\varphi_n\rag = 1$ such that  $\dot u(t)\in\calT_{u(t)}\calM$ and
\begin{align*}
\left\{ \begin{array}{ll}
\mathrm{Re}\,\langle v\,| \,\dot{u}(t)-\frac{1}{i\hbar}Hu(t)\rangle=0 & \text{(McLachlan principle)},\\*[2ex]
\mathrm{Re}\,\langle v\,| \,i\hbar\dot{u}(t)-Hu(t)\rangle=0 & \text{(Kramer--Saraceno principle)},
\end{array}\right.
\end{align*}
for all $v\in \mathcal{T}_{u(t)}\mathcal{M}$.
For both variational principles, the equations of motion for the single-particle functions depend on mean-field Hamiltonians,
which we define next.
\smallskip

\begin{definition}[Mean-field operator]
For $u=\varphi_1\cdots\varphi_N\in\calM$ and any linear operator $A:D(A)\to\calH$
whose domain contains $\calM$, we define the corresponding \emph{mean-field operator for the $n$-th particle}, $n=1,\ldots,N$, by
\[
A_n = \lag \psi_n\,|\,A\psi_n\rag
\]
as a linear operator with domain in $L^2(\R^3)$, where the inner product is over all variables except $x_n$ and $\psi_n = \prod\limits_{m\neq n}\varphi_m$ is the $n$-th single-hole function.
\end{definition}

We note that the mean-field operators $A_n$
are invariant under changes in the representation of $u$, since they depend quadratically on the single-hole
functions and are thus insensitive to gauge factors.
Moreover, their expectation values agree with the overall expectation value in the sense that
\[
\lag u\,|\, Au\rag = \lag \varphi_n\,|\, A_n\varphi_n\rag
\]
for all $n=1,\ldots,N$.
We obtain the following equations of motion, that are of the same form for both variational
principles, however, based on different time-dependent gauge factors:

\begin{theorem}[Equations of motion]\label{th31}
For initial data $\psi(x, 0)=\psi_0(x)=\varphi_1(x_1, 0)\ldots\varphi_N(x_N, 0)\in \mathcal{M}$, the
variational approximations given by the McLachlan and the Kramer--Saraceno principles
on the Hartree manifold~$\calM$ have the form
\[u(x_1, \ldots, x_N, t)=\varphi_1(x_1, t)\ldots\varphi_N(x_N,t),\]
where each single-particle function satisfies the mean field equation
\be\label{eqH}
i\hbar\dot{\varphi}_n(t)=
(H_n(t)+c_1(t)\delta_{1n}-E_0)\varphi_n(t)
\ee
with $H_n(t)$ the $n$-th mean-field Hamiltonian, $E_0(t)=\lag u(t)\,|\, Hu(t)\rag$
the total energy, and $c_1(t)=i\hbar\lag\varphi_1(t)\,|\,\dot\varphi_1(t)\rag$ a real-valued gauge factor.

\smallskip
For the Kramer--Saraceno principle, $E_0 = \lag \psi_0\,|\,H\psi_0\rag$ is the conserved energy. For the McLachlan principle, $E_0$ may be
time-dependent and is related to the gauge factor via
\be\label{mc}
E_0(t) = c_1(t).
\ee
\end{theorem}

\begin{proof} Firstly we write $\dot{u}$ as the unique representation $\dot{u}=\sum\limits_{m=1}^N\dot{\varphi}_m\psi_m$, where the time-derivatives of the single-particle functions satisfy the gauge conditions \eqref{eq:gauge}. We next fix $n$ and consider a tangent function $v=v_n\psi_n\in\calT_u\calM$ with $v_n\in L^2(\mathbb{R}^3)$ satisfying \eqref{eq:gauge} and calculate
\begin{align*}
\mathrm{Re}\lag v\,|\,i\hbar\dot{u}-Hu\rag&=
\mathrm{Re}\,\langle v_n\psi_n\,|\,i\hbar\sum\limits_{m=1}^N \dot{\varphi}_m\psi_m-H(\varphi_n\psi_n)\,\rangle\\
&=\mathrm{Re}\,\langle v_n\,|\,i\hbar\dot{\varphi}_n-H_n \varphi_n\rangle+
\mathrm{Re}\,(\lag v_n\,|\,\varphi_n\rag\sum\limits_{m\neq n} i\hbar\lag \varphi_m\,|\,\dot\varphi_m\rag)\\
&=\mathrm{Re}\,\langle v_n\,|\,i\hbar\dot{\varphi}_n-H_n \varphi_n\rangle
\end{align*}
where we have used the definition of the mean-field Hamiltonian $H_n$ and the gauge conditions,
which imply that $\lag v_n\,|\,\varphi_n\rag \lag \varphi_m\,|\,\dot\varphi_m\rag = 0$ for all $m\neq n$.
We first consider the case $n=1$. The Kramer--Saraceno orthogonality condition implies that
$\mathrm{Re}\,\langle v_1\,|\,i\hbar\dot{\varphi}_1-H_1 \varphi_1\rangle = 0$ for all $v_1\in L^2(\R^3)$ with
$\mathrm{Re}\lag v_1\,|\,\varphi_1\rag = 0$. Therefore, there exists $\kappa_1\in\R$ with
$i\hbar\dot{\varphi}_1-H_1 \varphi_1 = \kappa_1\varphi_1$. For the case $n\ge 2$, we have
$\mathrm{Re}\,\langle v_n\,|\,i\hbar\dot{\varphi}_n-H_n \varphi_n\rangle = 0$ for all $v_n\in L^2(\R^3)$ with
$\lag v_n\,|\,\varphi_n\rag = 0$. We write $i\hbar\dot{\varphi}_n-H_n \varphi_n =
\kappa_n\varphi_n + \tilde\varphi_n$ with $\kappa_n\in\C$ and $\lag \tilde\varphi_n\,|\,\varphi_n\rag = 0$.
Then,
\[
0 = \mathrm{Re}\,\langle \tilde\varphi_n\,|\,i\hbar\dot{\varphi}_n-H_n \varphi_n\rangle =
\mathrm{Re}\,\langle \tilde\varphi_n\,|\,\kappa_n\varphi_n + \tilde\varphi_n\rangle = \|\tilde\varphi_n\|^2
\]
and hence $\tilde\varphi_n=0$. In summary, for all $n\ge 1$ there exist constants $\kappa_n\in\C$ with
$i\hbar\dot{\varphi}_n-H_n \varphi_n = \kappa_n\varphi_n$.
The normalisation of the single particle functions implies that
\begin{align*}
\kappa_n&=\lag\varphi_n\,|\,\kappa_n\varphi_n\rag=
\lag\varphi_n\,|\,i\hbar \dot{\varphi}_n-H_n \varphi_n\rag=i\hbar\lag\varphi_n\,|\,\dot\varphi_n\rag-\lag u\,|\,Hu\rag =c_1\delta_{1n}-E_0,
\end{align*}
where $E_0=\lag u\,|\, Hu\rag=\lag \psi_0\,|\,H\psi_0\rag$ is the conserved energy. Thus \eqref{eqH} follows.

Now we turn to the McLachlan principle and write its orthogonality condition as
\begin{align}
0&=\mathrm{Re}\langle v\,|\,\dot{u}-\frac{1}{i\hbar}Hu\rangle=
\mathrm{Re}\,\langle v_n\psi_n\,|\,\sum\limits_{m=1}^N \dot{\varphi}_m\psi_m-\frac{1}{i\hbar}H(\varphi_n\psi_n)\,\rangle\nn\\
&=\mathrm{Re}\,\langle v_n\,|\,\dot{\varphi}_n-\frac{1}{i\hbar}H_n \varphi_n \rangle
+
\mathrm{Re}\,(\lag v_n\,|\,\varphi_n\rag\sum\limits_{m\neq n}\lag \varphi_m\,|\,\dot\varphi_m\rag)\nn\\
&=\mathrm{Re}\,\langle v_n\,|\,\dot{\varphi}_n-\frac{1}{i\hbar}H_n \varphi_n \rangle.\label{tmp1}
\end{align}
Hence, the orthogonality and the gauge conditions imply the existence of $\kappa_n\in\C$ with $\dot{\varphi}_n-\frac{1}{i\hbar}H_n \varphi_n =
\kappa_n\varphi_n$. Moreover, $\kappa_1\in\R$. We then obtain that
\begin{align*}
\kappa_1&=\mathrm{Re}\lag\varphi_1\,|\,\kappa_1\varphi_1\rag=\mathrm{Re}
\,\langle\varphi_1\,|\,\dot{\varphi}_1-\frac{1}{i\hbar}H_1 \varphi_1\rangle=0,\\
\kappa_n&=\lag\varphi_n\,|\,\kappa_n\varphi_n\rag=
\,\langle\varphi_n\,|\,\dot{\varphi}_n-\frac{1}{i\hbar}H_n \varphi_n\rangle=-\frac{1}{i\hbar}E_0,\quad n\ge 2.
\end{align*}
In view of Corollary \ref{coro}, in particular \eqref{eq:cond}, the McLachlan solution
has a time-dependent total energy satisfying
\[E_0=\lag u\,|\, Hu\rag=\lag u\,|\, i\hbar \dot{u}\rag=
\sum\limits_{m=1}^N\lag \varphi_m\psi_m\,|\,i\hbar \dot{\varphi}_m\psi_m\rag=c_1.
\]
We may therefore write $\kappa_n = \frac{1}{i\hbar}(c_1\delta_{1n}-E_0)$.
Thus for both principles, we have a unified form for each single-particle function
\[
i\hbar\dot{\varphi}_n=
(H_n+c_1\delta_{1n}-E_0)\varphi_n,
\]
which completes the proof.
\end{proof}

Due to the normalization constraint in the definition of the manifold $\calM$, both Hartree solutions $u(x,t) = \varphi_1(x_1,t)\cdots\varphi_N(x_N,t)\in\calM$, the McLachlan and the Kramer--Saraceno one, automatically satisfy
\[
i\hbar\,\langle u(t)\,|\,\dot u(t)\rangle = \sum_{n=1}^N i\hbar\, \langle\varphi_n(t)\,|\,\dot\varphi_n(t)\rangle
=c_1(t).
 \]
The gauge relation \eqref{mc}, that holds for the McLachlan solution, therefore yields
\[
i\hbar \,\langle u(t)\,|\,\dot u(t) \rangle = \lag u(t)\,|\,Hu(t)\rag,
\]
which is a variational analogue of the gauge property $i\hbar \,\langle\psi(t)\,|\,\dot\psi(t) \rangle = \lag \psi(t)\,|\,H\psi(t)\rag$,
that is satisfied by the solution $\psi(t)$ of the Schr\"odinger equation~\eqref{schro}.

\begin{remark}
If the Hartree manifold is defined without normalization constraint, but with a multiplicative
scalar factor, as
\[
\Big\{u\in L^2(\mathbb{R}^{3N})\mid\, u(x)=a\prod\limits_{n=1}^N\varphi_n(x_n)\neq0,\,\,\, a\in\C,\,\,\,
\varphi_n\in L^2(\mathbb{R}^3)\Big\},
\]
then one has complex-linear tangent spaces and thus naturally enters the setting for the Dirac--Frenkel
principle. The equations of motion are similar to the ones of Theorem~\ref{th31}. The complex-valued
factor $a(t)$ may absorb various gauge contributions, see \cite[Theorem~3.1]{Lu}.
\end{remark}

\subsection{The variational Hamiltonian}
Theorem~\ref{th31} and its equations of motion only use the self-adjointness of the Hamiltonian $H$. If the Hamiltonian allows for an additive splitting as in \eqref{Ham}, then
the $n$-th mean-field Hamiltonian $H_n = \lag \psi_n\,|\, H\psi_n\rag$ takes the specific form
\be\label{Hn}
H_n= \lag \psi_n\,|\,h_n\psi_n\rag + \sum_{m\neq n} \lag \psi_n\,|\,h_m\psi_n\rag + \lag\psi_n\,|\,V\psi_n\rag
= h_n + \sum_{m\neq n} \epsilon_m + V_n,
\ee
where $V_n:=V_n(x_n)=\lag \psi_n\,|\,V\psi_n\rag$
is the $n$-th mean-field potential and
\[
\epsilon_m=\lag\varphi_m\,|\,h_m\varphi_m\rag
\]
the $m$-th average one-particle energy. If the single-particle Hamiltonian $h_m$ is the Laplacian $-\frac{\hbar^2}{2M_m}\Delta_{x_m}$, then
\[
\epsilon_m=\frac{\hbar^2}{2M_m}\|\nabla_{x_m}\varphi_m\|^2.
\]
The sum of the mean-field Hamiltonians result in the variational Hamiltonian, that can be related to the
original one in terms of the zero-mean fluctuating potential.

\smallskip
\begin{theorem}[Variational Hamiltonian]\label{varHam}
If the Hamiltonian $H$ is of the form \eqref{Ham}, then the McLachlan variational solution given in Theorem~\ref{th31} satisfies a non-linear evolution equation
$i\hbar \dot{u}(t)=H_{u(t)} u(t)$ with right handside
\[
H_{u(t)}=\sum_{n=1}^N H_n(t)-(N-1)E_0.
\]
In particular, the difference of the true and the variational Hamiltonian is given by the
zero-mean fluctuating potential, $H-H_{u(t)} = \widetilde V(t)$, with
\[
\widetilde V(t) = V - \sum_{n=1}^N V_n(t) + (N-1)V_0(t),
\]
and $V_0(t) = \lag u(t)\,|\, Vu(t)\rag$ the potential energy with respect to $u(t)$.
\end{theorem}
\begin{proof}
We write the total energy as
\[
E_0=\lag u\,|\,Hu\rag=\sum_{n=1}^N\lag u\,|\, h_nu\rag + \lag u\,|\,Vu\rag = \sum_{n=1}^N\epsilon_n+V_0.
\]
Therefore, by the summation property \eqref{mc} of the gauge factors of the McLachlan principle, we get
\begin{align*}
H_{u}&=\sum_{n=1}^N \left(H_n+E_0(\delta_{1n}-1)\right) =\sum_{n=1}^N H_n-(N-1)E_0,\\
&= \sum_{n=1}^N(h_n+V_n) + (N-1)\big(\sum_{n=1}^N\epsilon_n - E_0\big)\\
&= \sum_{n=1}^N(h_n+V_n) - (N-1)V_0 = H -\widetilde V,
\end{align*}
which completes the proof.
\end{proof}

We note the zero-mean fluctuating potential indeed has zero mean with respect to the variational solution,
\be\label{V0u}
\big\langle u(t)\,|\,\widetilde{V}(t)u(t)\big\rangle = 0,
\ee
since
\begin{align*}
\big\langle u(t)\,|\,\widetilde{V}(t)u(t)\big\rangle &=
\lag u(t)\,|\,Vu(t)\rag+(N-1)V_0(t)-\sum\limits_{n=1}^N\lag u(t)\,|\, V_n(t)u(t)\rag\nn\\
&=NV_0(t)-\sum\limits_{n=1}^N\lag \varphi_n(t)\,|\,V_n(t)\varphi_n(t)\rag=0.
\end{align*}

\begin{remark}
Consider the special Hamiltonian $H = \sum\limits_{n=1}^N h_n + V$, where the potential $V$ separates in the sense that
\[
V(x) = \sum_{m=1}^N V^{(m)}(x_m)
\]
with each $V^{(m)}:\R^{3}\to\R$ a single particle potential. Then,
the mean field potentials satisfy
\begin{align*}
V_n(t) &= \sum_{m=1}^N \langle \psi_n(t)\,|\,V^{(m)}\psi_n(t)\rangle= V^{(n)} + \sum_{m\neq n} \langle \varphi_m(t)\,|\,V^{(m)}\varphi_m(t)\rangle\\
&= V^{(n)} - \langle \varphi_n(t)\,|\,V^{(n)}\varphi_n(t)\rangle + V_0(t)
\quad\text{for all}\ n=1,\ldots,N,
\end{align*}
and we therefore have $\widetilde V(t) = 0$.
\end{remark}

The zero-mean fluctuating potential $\widetilde V$ describes the local minimum distance that governs the a posteriori error, see also the discussion by R. Martinazzo and I. Burghardt in \cite{MB}.

\begin{corollary}\label{cor:zmfp}
For the Hamiltonian $H$ given in \eqref{Ham}, the McLachlan variational approximation $u(t)\in\calM$ given in Theorem \ref{th31} satisfies
\[
\varepsilon_\calM(u)=\frac{1}{\hbar}\|\widetilde{V}u\|=\frac{1}{\hbar}\ \Big(
\|Vu\|^2-\sum\limits_{n=1}^N\|V_n\varphi_n\|^2+(N-1)V_0^2\Big)^{1/2}.\]
\end{corollary}

\begin{proof}
By the definition of $\varepsilon_\calM(u)$ in \eqref{unif} and Theorem~\ref{varHam},
\[
\varepsilon_\calM(u) = \frac{1}{\hbar}\|H_uu - Hu\| = \frac{1}{\hbar}\|\widetilde Vu\|.
\]
It remains to calculate the norm as
\begin{align*}
&\quad\|\widetilde{V}u\|^2 =\Big\langle \big(V-\sum_{n=1}^N V_n\big)u\,|\,\widetilde{V}u\Big\rangle\\
&=\Big\|\big(V-\sum_{n=1}^N V_n\big)u\Big\|^2+(N-1)V_0\Big\langle \big(V-\sum_{n=1}^N V_n\big)u\,|\,u\Big\rangle\\
&=\|Vu\|^2-2\sum_{n=1}^N\lag Vu\,|\,V_nu\rag+\Big\|\sum_{n=1}^NV_nu\Big\|^2+(N-1)V_0^2\\
&\quad-(N-1)V_0\sum_{n=1}^N\lag V_nu\,|\,u\rag\\
&=\|Vu\|^2-\sum_{n=1}^N\|V_n\varphi_n\|^2+\sum_{n=1}^N
\sum_{m\neq n}\langle V_nu\,|\,V_mu\rangle-(N-1)^2V_0^2\\
&=\|Vu\|^2-\sum_{n=1}^N\|V_n\varphi_n\|^2+(N-1)V_0^2,
\end{align*}
where we have applied the formulae
\begin{align*}
\lag V_nu\,|\, u\rag &= \lag V_n\varphi_n\,|\, \varphi_n\rag = \lag Vu\,|\, u\rag = V_0,\\
\lag Vu\,|\, V_nu\rag &= \lag V\varphi_n\psi_n\,|\, V_n\varphi_n\psi_n\rag =
\lag V_n\varphi_n\,|\, V_n\varphi_n\rag = \|V_n\varphi_n\|^2,
\end{align*}
and
\[
\lag V_nu\,|\, V_mu\rag =
\left\{ \begin{array}{ll}
\lag V_n\varphi_n\psi_n\,|\, V_n\varphi_n\psi_n\rag=\|V_n\varphi_n\|^2, & m=n,\\*[1ex]
\lag V_n\varphi_n\psi_n\,|\, V_m\varphi_m\psi_m\rag
= V_0^2, & m\neq n.
\end{array}\right.
\]
\end{proof}

\section{Frozen Gaussian wave-packet dynamics}
\label{secfro}

\subsection{The frozen Gaussians' manifold}
For a small positive parameter $\delta>0$, we consider the Gaussian function
\[
\varphi_0(x)=(2\pi\delta^2)^{-d/4}\exp\!{\left(-\frac{|x|^2}{4\delta^2}\right)},\quad x\in\R^d.
\]
The placement of the parameter $\delta$ ensures that $|\varphi_0|^2$ is a probability density with
mean zero and covariance matrix $\delta^2\Id$,
\[
\|\varphi_0\| = 1,\quad \lag \varphi_0\,|\, x\varphi_0\rag = 0,\quad
\lag x_m\varphi_0\,|\,x_n\varphi_0\rag = \delta^2\delta_{mn}.
\]
We define the corresponding $\delta$-dependant \emph{ladder operators} as
\begin{align*}
&A=(A_1, \ldots, A_d)=\frac{\widehat{q}}{2\delta}+\frac{i\delta\widehat{p}}{\hbar} =
\frac{x}{2\delta}+ \delta \nabla,\\
&A^\dagger=(A^\dagger_1, \ldots, A^\dagger_d)=
\frac{\widehat{q}}{2\delta}-\frac{i\delta\widehat{p}}{\hbar} =
\frac{x}{2\delta} - \delta \nabla,
\end{align*}
where $\widehat{q}$ and $\widehat{p}$ are the position and momentum operators given by
$(\widehat{q}\varphi)(x)=x\varphi(x)$ and $(\widehat{p}\varphi)(x)=-i\hbar\nabla\varphi(x)$.
We observe that that the lowering operator~$A$ and the raising operator~$A^\dagger$ are (formally) adjoint to each other and its components fulfill the canonical communication relations,
\[
[A_k,A_\ell^\dagger] = [\partial_k,x_\ell] = \delta_{k,\ell}\quad\text{for all}\ k,\ell=1,\ldots,d.
\]
Moreover, $A$ annihilates the Gaussian in the sense that, $A\varphi_0 = 0$.
Following the Bargman formalism \cite{Barg}, we use the raising operator $A^\dagger$ to translate the
Gaussian in phase space. We denote
\[
z\cdot w = z_1w_1+ \cdots + z_dw_d
\]
for any complex vectors  $z,w\in\C^d$, but also for the vector of raising operators such that
$z\cdot A^\dagger = z_1A^\dagger_1 + \cdots + z_dA^\dagger_d$.

\begin{lemma}[Gaussian wave-packet]\label{upro} Let $\delta>0$. For $z=q+ip\in \mathbb{C}^d$ with $q,p\in\R^d$ we define
\[
u(x)=\exp\! \left(-|z|^2/2+z\cdot A^\dagger\right)\varphi_0(x).
\]
Then,
\[
u(x)=(2\pi\delta^2)^{-d/4}\exp\!{\Big(-\frac{1}{4\delta^2}|x-2\delta q|^2+\frac{i}{\delta} p \cdot (x-\delta q)\Big)}.
\]
In particular,
\begin{align}
&\|u\|=1,\quad \lag u\,|\,\widehat q u\rag=2\delta q,\quad \lag u\,|\,\widehat p u\rag=\hbar/\delta\, p,\\*[1ex]
&\lag x_mu\,|\,x_nu\rag=\delta^2(\delta_{mn}+4q_m q_n),\quad
\|xu\|^2=\delta^2(d+4|q|^2),\label{upro1}\\*[1ex]
& Au=zu,\,\,\,A^\dagger u=(x/\delta-z)u.\label{upro2}
\end{align}
\end{lemma}

\begin{proof} We view the function $u$ as the solution at time $t=1$ of the Cauchy problem
\[
\partial_t\Psi = \Big(-\frac12|z|^2 + z\cdot A^\dagger\Big)\Psi,\qquad \Psi(0) = \varphi_0.
\]
We have
\[
\Psi(t,x) = \exp\left(-\frac{t}{2}\,|z|^2-\frac{t^2}{4}\, z\cdot z+\frac{t}{2\delta} \,z\cdot x\right) \varphi_0(x-t \delta z),
\]
as can be verified by direct calculation, that is,
\[
\partial_t\Psi
= \Big( -\frac12|z|^2-\frac{t}{2} z\cdot z + \frac{1}{2\delta}z\cdot x -\delta z\cdot\nabla_x + \frac{t}{2}z\cdot z\Big)\Psi= \Big( -\frac12|z|^2 + z\cdot A^\dagger\Big)\Psi.\]
Therefore,
\[
u(x) = \exp\left(-\frac{1}{2}\,|z|^2-\frac{1}{4}\, z\cdot z+\frac{1}{2\delta} \,z\cdot x\right) \varphi_0(x-\delta z).
\]
Since the quadratic form in this exponential function can be rewritten as
\[
-\frac12|z|^2-\frac14 z\cdot z + \frac{1}{2\delta}z\cdot x
- \frac{1}{4\delta^2}(x-\delta z)\cdot(x-\delta z)= -\frac{1}{4\delta^2}|x-2\delta q|^2 + \frac{i}{\delta}p\cdot (x-\delta q),
\]
thn $u$ indeed has the claimed form. For the position expectation we get
\begin{align*}
\lag u\,|\widehat{q}u\rag&=(2\pi\delta^2)^{-\frac{d}{2}}\int_{\mathbb{R}^d}x\exp\left(
-\frac{1}{2\delta^2}|x-2\delta q|^2\right)\D x\\
&=2\delta q \int_{\mathbb{R}^d} |\varphi_0(x)|^2 \,\D x =2\delta q.\end{align*}
Since the gradient satisfies
\[
\nabla u(x)=u(x)\left(-\frac{x-2\delta q}{2\delta^2}+\frac{ip}{\delta}\right)=u(x)\left(-\frac{1}{2\delta^2}x+\frac{1}{\delta}z\right),
\]
we calculate the position expectation as
\begin{align*}
\lag u\,|\,\widehat p u\rag &= -i\hbar\lag u\,|\,\nabla u\rag
= -i\hbar\left(-\frac{1}{2\delta^2}\lag u\,|\,xu\rag + \frac{1}{\delta}z\lag u\,|\,u\rag\right)\\
&=\frac{-i\hbar}{\delta}\left(-q + z\right) = \frac{\hbar}{\delta}\,p.
\end{align*}
For the second moments we have
\begin{align*}
\lag x_mu\,|\,x_nu\rag&=(2\pi\delta^2)^{-\frac{d}{2}}\int_{\mathbb{R}^d}x_m x_n\exp\left(-\frac{1}{2\delta^2}|x-2\delta q|^2\right)\D x\\
&=\int_{\mathbb{R}^d}(x_m+2\delta q_m)(x_n+2\delta q_n)\,|\varphi_0(x)|^2\,\D x\\
&=4\delta^2q_m q_n+\int_{\mathbb{R}^d} x_mx_n|\varphi_0(x)|^2\, \D x\\*[1ex]
&=4\delta^2q_m q_n+\delta^2\delta_{mn}.
\end{align*}
In particular,
\[
\|xu\|^2=\sum\limits_{n=1}^d\lag x_n u\,|\, x_nu\rag=
\delta^2\sum\limits_{n=1}^d(1+4q_n^2)=\delta^2(d+4|q|^2).
\]
which completes the proof for \eqref{upro1}. Using the formula for $\nabla u$
once more, we obtain
\[
Au=\frac{1}{2\delta}xu+\delta\nabla u=zu,\quad
A^\dagger u=\frac{1}{2\delta}xu-\delta\nabla u=\big(\frac{x}{\delta}-z\big)u,
\]
which concludes \eqref{upro2}.
\end{proof}

Now, for a fixed width parameter $\delta>0$, we consider the variational \emph{frozen} Gaussian approximation on the manifold
\begin{align*}
\mathcal{M}&=\left\{u(\theta,q,p)\mid  \theta\in\mathbb{R},\ p, q\in\mathbb{R}^d\right\}\subset L^2(\R^d),
\end{align*}
where $u(\theta,q,p)$ denotes the wave-packet
\be\label{eq:frozen}
u(\theta, q, p)(x) =
(2\pi\delta^2)^{-\frac{d}{4}}\exp\!{\left(i\theta-\frac{1}{4\delta^2}|x-2\delta q|^2+\frac{i}{\delta} p \cdot (x-\delta q)\right)}.
\ee
Note that $\|u(\theta,q,p)\| = 1$ for all choices of the parameters $\theta\in\R$, $q,p\in\R^d$. The tangent spaces for this manifold can easily be calculated as follows:
\begin{lemma}[Tangent space]\label{lem:tang}
For $\forall u\in \calM$, the tangent space $\calT_u\calM$ is given by
\begin{align*}
\calT_u\calM
&=\left\{wu\,|\,w(x)=(a+ib)\cdot x-a\cdot 2\delta q+ic,\quad a,\,b\in\mathbb{R}^d,\,\,\, c\in\mathbb{R}\right\}\\
&= \left\{Bu\mid B = (a+ib)\cdot(A^\dagger - \frac{\overline{z}}{2})-\frac{z}{2}\cdot(a-ib) + ic,\,\, a,b\in\R^d,\, c\in\R\right\}.
\end{align*}
In particular, $\calT_u\calM$ is a real-linear $(2d+1)$-dimensional subspace of the Hilbert space $L^2(\R^d)$, that is not complex-linear, but $iu\in\calT_u\calM$.
\end{lemma}
\begin{proof} By definition, every tangent function $v\in \calT_u\calM$ is of the form
\begin{align*}
v&=\left[i\widetilde{\theta}+\frac{\widetilde{q}}{\delta}
\cdot (x-2\delta q)+\frac{i\widetilde{p}}{\delta}\cdot (x-\delta q)-ip\cdot \widetilde{q}\right]u\\
&=\left[\frac{1}{\delta}\left( \widetilde{q} + i\widetilde{p}\right)\cdot x
- 2\widetilde{q}\cdot q-i\left(\widetilde{p}\cdot q-p\cdot \widetilde{q}+\widetilde{\theta}\right)\right]u,
\end{align*}
where $\widetilde{\theta}\in \mathbb{R}$, $\widetilde{q},
\widetilde{p}\in \mathbb{R}^d$, which immediately gives that $v=wu$ with $w$ a linear polynomial of the claimed form. An analogous calculation provides the tangent functions within the Bargman formalism. Due to the restriction on the real part of the constant term of the polynomial $w$, the tangent space fails to be complex-linear. By choosing $a=b=0$, $c=1$ we have $iu\in\calT_u\calM$.
\end{proof}

\subsection{Structural properties of the variational Gaussians}
Since any $u\in\calM$ is normalized and satisfies $iu\in\calT_u\calM$, we can easily draw first conclusions on the properties of the variational solutions.
\smallskip
\begin{corollary}
Let $\delta>0$ and consider the manifold of frozen Gaussian wave packets $\calM$.
\begin{enumerate}
\item
For the McLachlan variational solution $u(t)$, the energies for the Hamiltonians $H$ and $H_{u(t)}$ coincide.
\item
The Kramer--Saraceno principle leaves the phase parameter $\theta(t)$ of the frozen
Gaussian $u(t) = u(\theta(t),q(t),p(t))\in\calM$ undetermined.
\end{enumerate}
\end{corollary}

\begin{proof}
For the McLachlan principle, we just observe that the frozen Gaussian
manifold $\calM$ satisfies the assumptions of Corollary~\ref{coro}.
Therefore, the energies for the Hamiltonians $H$ and $H_{u(t)}$ coincide.

Next we observe that any normalized function $u(t)$ satisfies
\be\label{eq:realiu}
\mathrm{Re}\lag iu(t)\,|\,i\hbar\dot u(t)-Hu(t)\rag = 0,
\ee
since $H$ is self-adjoint.
This implies for the Kramer--Saraceno solution $u(t)$, that on the one-dimensional subspace $\{i\alpha u(t)\mid \alpha\in\R\}$, the projection condition \eqref{symplectic_real} is automatically satisfied, which leaves one of its parameters undetermined.
\end{proof}

Using the explicit form of the tangent functions, we derive a further structural property of the McLachlan solution.
\smallskip
\begin{proposition}[Dirac--Frenkel condition]\label{prop:df}
Consider for $\delta>0$ the frozen Gaussian manifold $\calM$.
The McLachlan variational solution $u(t)$ satisfies the Dirac--Frenkel condition
\[
\langle v\,|\,\dot u(t)-\frac{1}{i\hbar}Hu(t)\rangle = 0,\quad\forall v\in\calT_{u(t)}\calM.
\]
In particular, the total energy $E_0 = \lag u(t)\,|\,Hu(t)\rag$ is independent of time.
\end{proposition}

\begin{proof}
Let $u=u(t)$ be the McLachlan solution.
We examine the inner product $\lag v\,|\,i\hbar\dot u-Hu\rag$ for the tangent functions
\[
v=iu, \ ix_mu, \ (x_m-2\delta q_m)u,\quad\text{with}\quad m=1,\ldots,d,
\]
since they span the tangent space at $u$. For $v=iu$, the combination of the McLachlan
condition~\eqref{metric} with the above \eqref{eq:realiu} yields
\[\label{org1}
\lag iu\,|\, i\hbar\dot{u}-Hu\rag = 0.
\]
This equation in turn implies for the other basis functions that
\[
\lag ix_m u\,|\, i\hbar\dot{u}-Hu\rag = -i\lag (x_m-2\delta q_m) u\,|\, i\hbar\dot{u}-Hu\rag.
\]
Therefore,
\begin{align*}
&\mathrm{Re}\lag ix_m u\,|\, i\hbar\dot{u}-Hu\rag =
\mathrm{Im}\lag (x_m-2\delta q_m) u\,|\, i\hbar\dot{u}-Hu\rag,\\
&\mathrm{Im}\lag ix_m u\,|\, i\hbar\dot{u}-Hu\rag =-
\mathrm{Re}\lag (x_m-2\delta q_m) u\,|\, i\hbar\dot{u}-Hu\rag,
\end{align*}
and we obtain
\[
\mathrm{Re}\lag v\,|\,i\hbar\dot u - Hu\rag = 0,\quad\forall v\in\calT_u\calM.
\]
Hence, both the imaginary and the real part of the inner product vanish on the tangent space,
which amounts to the claimed Dirac--Frenkel property.
\end{proof}

We can also easily describe self-adjoint Hamiltonians
$H$, for which the variational approximation are exact.
\smallskip
\begin{proposition}[Exactness]\label{pro31}
If the Hamiltonian can be written as
\[
H=\lambda A^\dagger\cdot A+\mu\cdot \widehat{q}+\omega= \lambda\left(-\delta^2\Delta + \frac{1}{4\delta^2}|x|^2 - \frac{d}{2}\right) +\mu\cdot x+\omega,\]
where $\lambda, \omega\in\mathbb{R}$, $\mu\in \mathbb{R}^d$, then the variational approximation invoked by the McLachlan principle is exact, i.e., $u(t)=\psi(t)$,
provided that $\psi(0)\in \calM$. If the phase parameter of the Kramer--Saraceno solution is chosen as for the McLachlan one, then also the Kramer--Saraceno solution is exact.
\end{proposition}

\begin{proof}
We calculate
\begin{align*}
A^\dagger\cdot A
&= \left(\frac{1}{2\delta}x-\delta\nabla\right)\cdot \left(\frac{1}{2\delta}x+\delta\nabla\right)= -\delta^2\Delta +\frac{1}{4\delta^2}|x|^2 + \frac12\left( x\cdot\nabla-\nabla\cdot x\right)\\
&= -\delta^2\Delta +\frac{1}{4\delta^2}|x|^2  -\frac{d}{2},
\end{align*}
since $\nabla\cdot x = \sum\limits_m \partial_mx_m = d + x\cdot\nabla$.
Using the relations \eqref{upro2}, we have for fixed $u\in\calM$
\begin{align*}
Hu&=\lambda z\cdot (\frac{x}{\delta}-z)u+\mu\cdot xu+\omega u\\
&=\big(\frac{\lambda}{\delta}q+\mu\big)\cdot x u+\left(\omega-\lambda(|q|^2-|p|^2)\right)u+
\frac{i\lambda}{\delta}p\cdot (x-2\delta q)u,
\end{align*}
which implies that $\frac{1}{i\hbar}Hu = wu\in\calT_u\calM$ for a linear polynomial $w$ of the form
\[
w(x)=(a+ib)\cdot x-a\cdot 2\delta q+ic
\]
with $a=\frac{\lambda}{\hbar\delta}\,p,\quad
b=-\frac{1}{\hbar}\left(\frac{\lambda}{\delta}\,q+\mu\right),\quad
c=\frac{\lambda}{\hbar}(|q|^2-|p|^2)-\frac{\omega}{\hbar}$.
Recalling the orthogonality condition \eqref{metric} for the McLachlan principle, or the symplectic formulation \eqref{symplectic} for the Kramer-Saraceno principle, we get that $\dot{u}= \frac{1}{i\hbar} Hu$, which immediately gives $u(t)=\psi(t)$.
\end{proof}

\subsection{Equations of motion and energy fluctuations}
We next derive equations of
motion for both variational principles that are formulated in terms of the ladder operator $A$ and the Hamiltonian $H$.
\begin{theorem}[Equations of motion]\label{paeq}
Let $\delta>0$ and consider an initial datum $u(0) = \psi_0\in\calM$.
The McLachlan variational principle
determines a wave-packet $u(t)$ of the form \eqref{eq:frozen},
where the parameters $\theta(t)\in\mathbb{R}$, $z(t)=q(t)+ip(t)\in\mathbb{C}^d$ satisfy
\begin{align}
&\hbar \dot{\theta}(t)=\mathrm{Re}\left(\overline{z}(t)\cdot \lag  u(t)\,|\,[H,A]u(t)\rag\right) - E_0,
\label{theteq}\\*[1ex]
&\hbar \dot z(t) = i\lag u(t)\,|\,[H,A]u(t)\rag.\label{zeq}
\end{align}
The Kramer--Saraceno principle uniquely determines the parameter $z(t)$ to satisfy \eqref{zeq}, but leaves the phase $\theta(t)$ of the wave-packet undetermined.
For both variational principles, the total energy $E_ 0 = \lag u(t)\,|\,Hu(t)\rag$ is independent of time.
\end{theorem}

\begin{proof} Essentially repeating the calculation of Lemma~\ref{lem:tang}, we may write the time derivative of a frozen Gaussian wave packet as
\begin{align}
\dot{u}&=\big[i\dot{\theta}+\frac{\dot{q}}{\delta}
\cdot (x-2\delta q)+\frac{i\dot{p}}{\delta}\cdot (x-\delta q)-ip\cdot \dot{q}\big]u\nn\\
&=\big[i\dot{\theta}+\frac{\dot{z}}{\delta}
\cdot (x-\delta q)-\dot{q}\cdot z\big]u\label{du}.
\end{align}
We now evaluate the inner product
$\lag v\,|\,i\hbar\dot u-Hu\rag$ for different tangent functions $v\in\calT_u\calM$. For $v=iu\in \calT_u\calM$, we get
\begin{align}\nonumber
\lag iu\,|\, i\hbar\dot{u}-Hu\rag&=\hbar\langle u\,|\,\big(i\dot{\theta}-\dot z\cdot q - \dot q\cdot z\big)u\rangle
+\frac{\hbar}{\delta}\,\dot z\cdot\langle u\,|\,xu\rangle+i\lag u\,|\, Hu\rag\\\nonumber
&=\hbar\big(i\dot{\theta}-\dot{z}\cdot q-\dot{q}\cdot z\big)+2\hbar\dot{z}\cdot q+iE_0\\\label{eq:iu}
&=i\left(\hbar \dot{\theta}+\hbar\dot{p}\cdot q-\hbar\dot{q}\cdot p+E_0\right),
\end{align}
where we have used the position expectation $\lag u\,|\,xu\rag = 2\delta q$.

\medskip
We now focus on the McLachlan solution. The orthogonality condition~\eqref{metric} implies for the phase parameter
\[
\hbar \dot{\theta}=-\hbar\dot{p}\cdot q+\hbar\dot{q}\cdot p-E_0.
\]
We next eliminate $\dot\theta$ in the expression for $\dot u$ given in \eqref{du}. We obtain that
\be\label{du2}
\dot{u}= \big[-i\dot{p}\cdot q+i\dot{q}\cdot p-\frac{i}{\hbar}E_0+\frac{\dot{z}}{\delta}
\cdot (x-\delta q)-\dot{q}\cdot z\big]u
= \big[-\frac{i}{\hbar}E_0+ \dot{z} \cdot \big(\frac{x}{\delta}-2q\big)\big]u.
\ee
Using the Dirac--Frenkel property of Proposition~\ref{prop:df}, we work with the
tangent functions $v=ix_mu$ and conditions
\[\label{org2}
\lag ix_m u\,|\, i\hbar\dot{u}-Hu\rag = 0,\quad m=1,\ldots,d,
\]
for deriving the equations of motion for the $z$-parameter.
We insert the expression for $\dot u$ obtained in equation~\eqref{du2} and calculate
\begin{align*}
&\quad\lag ix_m u\,|\, i\hbar\dot{u}-Hu\rag\nn\\
&=-iE_0\lag x_m u\,|\,u\rag +\frac{\hbar}{\delta} \dot{z}\cdot\lag x_m u\,|\,xu\rag
-2\hbar \dot z\cdot q\lag x_mu\,|\, u\rag+i\lag x_m u\,|\, Hu\rag\nn\\
&=2\delta q_m \left(-iE_0-2\hbar \dot z\cdot q\right)+\hbar\delta \left(\dot z_m + 4 q_m \dot z\cdot q\right)+i\delta\langle (z_m+A_m^\dagger)u\,|\, Hu\rangle\nn\\*[1ex]
&=-2i\delta q_mE_0+\hbar \delta\dot{z}_m+i\delta \overline z_mE_0+i\delta \lag u\,|\, A_mHu\rag\nn\\*[1ex]
&=\hbar\delta\dot{z}_m-i\delta E_0 z_m+i\delta\lag  u\,|\, A_mHu\rag\\*[1ex]
&=\hbar\delta\dot{z}_m - i\delta\lag u\,|\,[H,A_m]u\rag,
\end{align*}
where we have used $x_m u= \delta (z_m+A^\dagger_m)u$, and
\[
\lag x_mu\,|\,x_nu\rag =
\delta^2 \left(\delta_{mn} + 4 q_m q_n\right),\quad
E_0z_m = \lag u\,|\,Hz_m u\rag = \lag u\,|\,HA_mu\rag.
\]
Setting the inner product to zero, we obtain the claimed equation \eqref{zeq}. Inserting this information in the equation for
$\theta$, we derive
\[
\hbar \dot{\theta}
= -\mathrm{Re}\lag u\,|\,[H,A]u\rag\cdot q -\mathrm{Im}\lag u\,|\,[H,A]u\rag\cdot p - E_0= \mathrm{Re}\left(\overline{z}\cdot \lag  u\,|\,[H,A]u\rag\right) - E_0.
\]
which is equation~\eqref{theteq}.

\medskip
Now we turn to the Kramer--Saraceno principle. Without any assumption on the phase $\theta$, we calculate
\begin{align*}
&\quad \lag ix_mu\,|\,i\hbar\dot u - Hu\rag\\
&=
\langle ix_mu\,|\,i\hbar \big[i\dot{\theta}+\frac{\dot{z}}{\delta}
\cdot (x-\delta q)-\dot{q}\cdot z\big]u\rangle +i \delta\langle(z_m+A^\dagger_m)u\, |\,Hu\rangle\\
&=
\hbar\big[i\dot{\theta}-\dot{z}\cdot q-\dot{q}\cdot z\big]\,2\delta q_m +
\hbar\delta \left(\dot z_m + 4 q_m \dot z\cdot q\right)
+i \delta\langle(z_m+A^\dagger_m)u\, |\,Hu\rangle\\*[1ex]
&=
\hbar\big[i\dot{\theta}+\dot{z}\cdot q-\dot{q}\cdot z\big]\,2\delta q_m +
\hbar\delta \dot z_m +i \delta\langle(z_m+A^\dagger_m)u\, |\,Hu\rangle.
\end{align*}
The orthogonality condition \eqref{symplectic_real} requires the real part vanishes, that is
\[
0 = \hbar\dot q_m - \mathrm{Im}\langle (z_m+A_m^\dagger)u\,|\,Hu\rangle.
\]
Since $H$ is self-adjoint, we have
\begin{align*}
\langle (z_m+A_m^\dagger)u\,|\,Hu\rangle&=
\langle HA_mu\,|\,u\rangle + \langle u\,|\,A_mHu\rangle\\
&= -\langle u\,|\,[H,A_m]u\rangle + 2\,\mathrm{Re}\langle u\,|\,HA_m\rangle,\end{align*}
and $0 = \hbar\dot q_m + \mathrm{Im}\langle u\,|\,[H,A_m]u\rangle$,
which is the real part of equation \eqref{zeq}. Now, we consider the tangent functions $v=(x_m-2\delta q_m) u\in \calT_u\calM$. Using equation \eqref{eq:iu}, we eliminate $\dot\theta$ and calculate
\begin{align*}
&\quad\mathrm{Re}\lag (x_m-2\delta q_m)u\,|\,i\hbar\dot u - Hu\rag\\
&=
\mathrm{Re}\lag x_mu\,|\,i\hbar\dot u -Hu\rag - 2\delta q_m\,\mathrm{Re}\lag u\,|\,i\hbar\dot u -Hu\rag\\
&= -\mathrm{Im}\lag ix_mu\,|\,i\hbar\dot u -Hu\rag+2\delta q_m\mathrm{Im}\lag iu\,|\,i\hbar\dot u -Hu\rag\\
&= \delta(-\hbar \dot p_m -\mathrm{Re}\langle(z_m+A^\dagger_m)u\, |\,Hu\rangle+2q_mE_0).
\end{align*}
Since
\begin{align*}
&\quad 2q_mE_0 -\mathrm{Re}\langle(z_m+A^\dagger_m)u\, |\,Hu\rangle\\
&=
2\mathrm{Re}\langle u\,|\,HA_mu\rangle - \mathrm{Re}\langle HA_mu\, |\,u\rangle
-\mathrm{Re}\langle u\, |\,A_mHu\rangle\\*[1ex]
&= \mathrm{Re}\langle u\, |\,[H,A_m]u\rangle,
\end{align*}
the Kramer--Saraceno principle requires that $0 =\hbar \dot p_m -\mathrm{Re}\,\langle u\, |\,[H,A_m]u\rangle$, which is the imaginary part of equation \eqref{zeq}.
\end{proof}

Revisiting the previous proof, we may write the evolution of the McLachlan solution as
\begin{align*}
i\hbar \dot u
&= E_0 u(t) + i\hbar\, \dot z(t)\cdot (\frac{x}{\delta}-2q(t))u(t)\\
&= E_0 u(t)  +\lag u(t)\,|\,[A,H]u(t)\rag\cdot(A^\dagger-\overline z(t)) u(t).
\end{align*}
Since $(A-z)u = 0$ and $\overline{\lag u\,|\,[A,H]u\rag} = \lag [A,H]u\,|\,u\rag = -\langle u\,|\,[A^\dagger,H]u\rangle$,
we may symmetrise the above equation and write $i\hbar\dot u(t) = H_{u(t)}u(t)$  with a variational Hamiltonian of the form
\be\label{eq:GaussHam}
H_{u} =  E_0  +\langle u\,|\,[A,H]u\rangle\cdot(A^\dagger-\overline z)
-\langle u\,|\,[A^\dagger,H]u\rangle \cdot(A-z).
\ee

As pointed out by R. Martinazzo and I. Burghardt in \cite{MB},
the energy fluctuation of the variational Hamiltonian $H_{u(t)}$ has a simple description in terms of
the time derivative $\dot z(t)$.

\begin{corollary} The variational solution $u(t)$ as given in Theorem~\ref{paeq} defines a variational Hamiltonian
with energy fluctuation
\[
\|(H_{u(t)}-E_0)u(t)\| = \hbar\, |\dot z(t)|.
\]
This implies in particular for the minimal distance $\varepsilon_\calM(u(t))$ that
\be\label{local-g}
\varepsilon^2_\calM(u(t)) = \frac{1}{\hbar^2}
\Big(\|(H-E_0)u(t)\|^2-\hbar^2|\dot{z}(t)|^2\Big).
\ee
\end{corollary}

\begin{proof} We write $(H_u-E_0)u = i\hbar\, \dot z\cdot\left( x/\delta-2q\right) u$
and calculate the norm as
\[
\left\|\dot z \cdot\left( \frac{x}{\delta}-2q\right)u\right\|^2=
\frac{1}{\delta^2}\|(\dot z\cdot x) u\|^2 - \frac{4}{\delta}\,\mathrm{Re}\lag (\dot z\cdot x)u\,|\,(\dot z\cdot q) u\rag + 4|\dot z\cdot q|^2.
\]
Applying \eqref{upro1}, we have
\begin{align*}
&\|(\dot{z}\cdot x)u\|^2=\sum_{m,n}
\overline{\dot{z}_m}\dot{z}_n\lag x_m u\,|\,x_nu\rag=\delta^2|\dot{z}|^2+4\delta^2|\dot z\cdot q|^2,\\
&\lag (\dot z\cdot x)u\,|\,(\dot z\cdot q) u\rag =
(\dot z\cdot q) \sum_m \overline{\dot z_m} \lag x_mu\,|\,u\rag= 2\delta |\dot z\cdot q|^2.\end{align*}
Hence by definition, we get $\|(H_u-E_0)u\|^2 = \hbar^2 |\dot z|^2 +(4-8+4)\hbar^2 |\dot z\cdot q|^2 = \hbar^2|\dot z|^2$ and the proof is completed.
\end{proof}

\subsection{Schr\"odinger dynamics for frozen Gaussians}
We consider a Schr\"odinger operator
\[
H=\frac{\widehat{p}^2}{2}+V=-\frac{\hbar^2}{2}\Delta +V
\]
and reformulate the equations of motion of Theorem~\ref{paeq} such that they resemble those
of classical mechanics, up to an averaging of the potential's gradient.

\begin{corollary}\label{th34}
For the Schr\"odinger Hamiltonian $H=-\frac{\hbar^2}{2}\Delta +V$, the variational equations of motion for the
parameters $q$ and $p$, as given in Theorem~\ref{paeq}, can be rewritten as
\be\label{zd}
\dot q = \frac{\hbar}{2\delta^2}\,p,\qquad \dot p = -\frac{\delta}{\hbar}\lag u\,|\,(\nabla V)u\rag.
\ee
\end{corollary}

\begin{proof}
We calculate the commutator,
\[
[A,H]= \left[\tfrac{1}{2\delta}x+\delta\nabla,-\tfrac{\hbar^2}{2}\Delta + V\right]= -\frac{\hbar^2}{4\delta}[x,\Delta] + \delta [\nabla,V]
= \frac{\hbar^2}{2\delta}\nabla + \delta\nabla V,
\]
and average it over the Gaussian wave packet,
\[
\lag u\,|\,[A,H]u\rag
= \frac{\hbar^2}{2\delta}\lag u\,|\,\nabla u\rag + \delta\lag u\,|(\nabla V)u\rag= \frac{\hbar^2}{2\delta^2}ip + \delta\lag u\,|(\nabla V)u\rag,
\]
which proves \eqref{zd}.
\end{proof}

In the Schr\"odinger case, we can relate the approximation error to the quadratic
remainder of the potential $V$, when expanded around the center of the frozen wave packet, as follows:

\begin{theorem}\label{th35}
We consider a Schr\"odinger operator $H = -\frac{\hbar^2}{2}\Delta + V$.
Suppose that the potential $V$ is a smooth function such that there exists a constant $C_2>0$ with
\[
|\partial^\alpha V(x)| \le C_2\quad \text{for all}\quad  |\alpha| = 2, \ x\in\R^d.
\]
For $q\in\R^d$, denote by $W_{2,q}$ the quadratic remainder of the Taylor expansion of the potential $V$ around the point
$2\delta q$.
Then, the McLachlan variational solution~$u(t)$ satisfies
\be\label{localm}
\varepsilon_\calM(u(t))\le \frac{\hbar}{8\delta^2}\sqrt{d^2+2d} + \frac{1}{\hbar}\|W_{2,q(t)}u(t)\|\le
\left( \frac{\hbar}{8\delta^2} + \frac{\delta^2}{2\hbar}C_2\right) \sqrt{d^2+2d}.
\ee
The Kramer--Saraceno solution satisfies the same estimate, if its phase parameter is
chosen as for the McLachlan solution.
\end{theorem}

\begin{proof}
We linearize the potential $V(x)$ around the point $2\delta q$ and denote the quadratic remainder by $W_{2,q}(x)$,
\[
V(x) = V(2\delta q) + \nabla V(2\delta q)\cdot (x-2\delta q) + W_{2,q}(x).
\]
We then write the Schr\"odinger operator as $H= H_e + R_q(x)$,
where the exact Hamiltonian
\[
H_e =
\frac{\hbar^2}{2\delta^2}\left( -\delta^2\Delta +\frac{1}{4\delta^2}|x|^2 - \frac{d}{2}\right)
+ \mu\cdot x + \omega
\]
is defined with
\[
\mu = \nabla V(2\delta q) -\frac{\hbar^2}{2\delta^3}q\quad\text{and}\quad
\omega = \frac{\hbar^2d}{4\delta^2}+V(2\delta q) - \nabla V(2\delta q)\cdot 2\delta q + \frac{\hbar^2}{2\delta^2}|q|^2,
\]
while the remainder is given by $R_q(x) = - \frac{\hbar^2}{8\delta^4}|x-2\delta q|^2 +W_{2,q}(x)$.
By Proposition~\ref{pro31}, we have $\frac{1}{i}H_e u\in\calT_u\calM$ and therefore $iP^g_u(\frac{1}{i} H_eu) = H_e u$.
We thus obtain for the difference of the Hamiltonians that
\[
Hu-H_uu= Hu - iP^g_u(\tfrac{1}{i}Hu)= i(\Id-P^g_u)(\tfrac{1}{i}R_qu).
\]
This implies the estimate
\[
\varepsilon_\calM(u) = \frac{1}{\hbar}\|Hu - H_uu\| \le \frac{1}{\hbar}\|R_qu\|.
\]
Recalling that $|R_q(x)| \le \left( \frac{\hbar^2}{8\delta^4} + \frac12 C_2\right)|x-2\delta q|^2$, and
\[\||x-2\delta q|^2u\|^2
= (2\pi\delta^2)^{-d/2} \int_{\R^d} |y|^4 \exp(-\tfrac{1}{2\delta^2}|y|^2) \,\D y= \delta^4\cdot (d^2+2d),\]
we immediately get \eqref{localm} and the proof is completed.
\end{proof}

From the error estimate \eqref{localm} we find that the best choice for $\delta$ is of the form $\delta\sim \hbar^{1/2}$, since then the two additive contributions $\hbar/\delta^2$ and $\delta^2/\hbar$ are balanced.
However, we also note, that then the error is then just bounded in the semi-classical limit $\hbar\to0$.

\subsection{Comparion with thawed Gaussians}
Thawed Gaussian wave packets also depend on a width matrix $C\in\C^{d\times d}$, that is a complex
symmetric matrix, $C=C^T$, and has a positive definite imaginary part. They can be parametrised as
\[
\exp\Big(\frac{i}{\hbar}\Big(\frac12(x-q)\cdot C(x-q) + p\cdot(x-q) + \theta\Big)\Big),
\]
see \cite{FL}, \cite[II.4]{Lu}, \cite[\S3]{LL}. The manifold of such thawed Gaussian
wave packets has complex linear tangent spaces and thus naturally enters the framework of the
Dirac--Frenkel variational principle.
\begin{remark}
In the parametrization of the frozen wave packet \eqref{eq:frozen}, that we have used here, the prefactor
$-\tfrac{1}{4\delta^2}$ of the quadratic term corresponds to the matrix $-\tfrac{1}{2\hbar}\mathrm{Im}C$
in the above thawed ansatz. In the semi-classical limit regime, where $\hbar>0$ is considered as a small parameter tending to zero, one often assumes that the eigenvalues of $\mathrm{Im}C$ are bounded from below by a fixed spectral parameter $\rho>0$. An analogous assumption for the frozen wave packet is an upper bound of the form $\delta<\sqrt{\hbar/\rho}$.
\end{remark}

Let $W_{3,q(t)}(x)$ denote the cubic remainder of the Taylor expansion
of the potential $V(x)$ around the position center of the thawed variational wave packet $u_{\rm th}(t)$.
The proof of \cite[Theorem~4.4]{Lu} reveals that
\[
\varepsilon_{\calM_{\rm th}}(u_{\rm th}(t)) \le \frac{1}{\hbar}\|W_{3,q(t)}u_{\rm th}(t)\|,
\]
which is the analogous error estimate to the frozen one of Theorem~\ref{th35}.
Let $C_3>0$ be an upper bound on the cubic derivatives of the potential $V(x)$ and
$\rho>0$ a lower bound on the eigenvalues of the imaginary part of the width matrix $C(s)$ for $s\in[0,t]$. Then,
the above estimate implies
\[
\varepsilon_{\calM_{\rm th}}(u_{\rm th}(t)) \le C_3 \sqrt{\hbar/\rho^3}\ ,
\]
see \cite[Lemma~3.8]{LL}.
Hence, if the spectral parameter $\rho>0$ is considered as a fixed positive number, then the thawed Gaussian
have an accuracy of order $\sqrt\hbar$. In particular, the approximation error converges to zero in the semi-classical limit
$\hbar\to0$. However, in computational practice, thawed Gaussians occasionally spread in such a way,
that the spectral bound $\rho>0$ has to be considered as an additional small parameter, that
deteriorates the thawed Gaussians' approximation accuracy. In this regime, frozen and thawed
Gaussians are competing approaches. The literature shows successful applications of both of them,
see for example \cite{IJ,R,VB}.

\section{Conclusion}
We have characterized the McLachlan and the Kramer--Saraceno variational principles in terms of a metric and a symplectic orthogonality condition, respectively,
\[
\forall v\in\calT_{u(t)}\calM: \ \left\{ \begin{array}{ll}
\mathrm{Re}\langle v\,|\, \dot u(t) - \frac{1}{i\hbar}Hu(t)\rangle = 0 & \text{(McLachlan),}\\*[1ex]
\mathrm{Im}\langle v\,|\, \dot u(t) - \frac{1}{i\hbar}Hu(t)\rangle = 0 & \text{(Kramer--Saraceno).}
\end{array} \right.
\]
For both principles, norm conservation is possible under conditions on the manifold, namely,
$u\in\calT_{u}\calM$ (McLachlan) respectively $iu\in\calT_u\calM$ (Kramer--Saraceno) for all $u\in\calM$.
The Kramer--Saraceno variational solution is always energy conserving, while the McLachlan one satisfies
an a posteriori error estimate
\[
\|\psi(t)-u(t)\| \le \int_0^t \Big(\frac{1}{\hbar^2}\|Hu(s)\|^2 - \|\dot u(s)\|^2\Big)^{1/2} \ \D s,
\]
that allows for practical numerical implementations.

The time-dependent Hartree approximation has been formulated with a manifold of normalized
single particle functions. The resulting McLachlan and Kramer--Saraceno solutions are obtained
via the usual mean-field equations,
\[
i\hbar\dot{\varphi}_n(t)=
(H_n(t)+c_n(t)-E_0(t))\varphi_n(t),
\]
however, differ in terms of the relations satisfied by the
time-dependent gauge factors $c_n(t)$ and the total energy $E_0(t)$, that might be time-dependent
for the McLachlan solution. The above a posteriori error estimate is characterised
by the zero mean fluctuating potential.

For the manifold of frozen Gaussian wave packets, the Kramer--Saraceno principle leaves the phase parameter
of the variational wave packet undetermined, but otherwise results in the same equations of motion
for the position and the momentum center as produced by the McLachlan principle. For Schr\"odinger
Hamiltonians, we derive an a posteriori error estimate that
depends on the second order derivatives of the potential function.



\begin{thebibliography}{s99}


\bibitem{Barg} V. Bargmann, \emph{On a Hilbert space of analytic functions and an associated integral transform part I},
    Commun. Pure Appl. Math. 14 (1961), 187--214.


\bibitem{De} F. Deutsch, \emph{Best Approximation in Inner Product Spaces}, Springer, 2001.


\bibitem{FL}
E. Faou, C. Lubich, \emph{A Poisson integrator for Gaussian wave packet dynamics},
Comput. Visual. Sci. 9 (2006), 45--55.

\bibitem{DF} J. Frenkel, \emph{Wave Mechanics, Advanced General Theory}, Clarendon Press, Oxford, 1934.

\bibitem{GL}
F. Gatti, B. Lasorne, H.-D. Meyer, A. Nauts,
\emph{Applications of Quantum Dynamics in Chemistry},
Lectures Notes in Chemistry, Springer, 2017.



\bibitem{HG}
L. Hackl, T. Guaita, T. Shi, J. Haegeman, E. Demler, J. Cirac,
\emph{Geometry of variational methods: dynamics of closed quantum systems},
SciPost Phys. 9 (2020), 048.

\bibitem{IJ}
A.~Izmaylov, L. Joubert-Doriol,
\emph{Quantum Nonadiabatic Cloning of Entangled Coherent States},
J. Phys. Chem. Lett. 8 (2017), 1793--1797.

\bibitem{TD} P. Kramer, M. Saraceno, \emph{Geometry of the Time-Dependent Variational Principle in Quantum Mechanics}, Lecture Notes in Physics 140, Springer, Berlin, 1981.

\bibitem{Lu05} C. Lubich, \emph{On variational approximations in quantum molecular dynamics}, Math. Comput. 74 (2005), 765--779.

\bibitem{Lu} C. Lubich, \emph{From Quantum to Classical Molecular Dynamics:
Reduced Models and Numerical Analysis}, Zurich
Lectures in Advanced Mathematics, European Mathematical
Society, 2008.

\bibitem{LL}
C. Lasser, C. Lubich,
\emph{Computing quantum dynamics in the semiclassical regime},
Acta Numer. 29 (2020), 229--401.

\bibitem{MB}
R. Martinazzo and I. Burghardt,
\emph{Local-in-time error in variational quantum dynamics},
Phys. Rev. Lett. 124 (2020), 150601.

\bibitem{Mc} A. D. McLachlan, \emph{A variational solution of the time-dependent Schr\"odinger equation}, Mol. Phys. 8 (1964), 39-44.

\bibitem{R}
G. W. Richings, I. Polyak, K. E. Spinlove, G. A. Worth, I. Burghardt, B. Lasorne,
\emph{Quantum dynamics simulations using Gaussian wavepackets: the vMCG method},
Int. Rev. Phys. Chem. 34 (2015), 269--308.


\bibitem{VB}
J. Vanicek, T. Begusic,
\emph{Ab initio semiclassical evaluation of vibrationally resolved electronic spectra with thawed Gaussians},
Molecular Spectroscopy and Quantum Dynamics (2021), Elsevier, 199--229.


\end{thebibliography}
\end{document}